\documentclass[conference]{IEEEtran}
\usepackage{amsfonts}
\usepackage{times}
\usepackage{graphicx}
\usepackage{latexsym}
\usepackage{dsfont}
\usepackage{amssymb}
\usepackage{amsmath}
\usepackage{cite}
\usepackage{verbatim}
\usepackage{subfigure}
\newtheorem{theorem}{Theorem}

\newtheorem{algorithm}[theorem]{Algorithm}
\newtheorem{assumption}[theorem]{Assumption}

\newtheorem{lemma}[theorem]{Lemma}

\newtheorem{proposition}[theorem]{Proposition}

\newenvironment{proof}{ \textbf{Proof:} }{ \hfill $\Box$}

\newcommand{\figref}[1]{{Fig.}~\ref{#1}}

\def\bb0{{\mathbb{0}}}

\def\ba{{\mathbf{a}}}
\def\bb{{\mathbf{b}}}

\def\bff{{\mathbf{f}}}
\def\bg{{\mathbf{g}}}
\def\bh{{\mathbf{h}}}

\def\bn{{\mathbf{n}}}

\def\br{{\mathbf{r}}}
\def\bs{{\mathbf{s}}}

\def\bv{{\mathbf{v}}}
\def\bw{{\mathbf{w}}}
\def\bx{{\mathbf{x}}}

\def\bz{{\mathbf{z}}}
\def\b0{{\mathbf{0}}}

\def\bA{{\mathbf{A}}}

\def\bD{{\mathbf{D}}}

\def\bF{{\mathbf{F}}}

\def\bH{{\mathbf{H}}}
\def\bI{{\mathbf{I}}}

\def\bP{{\mathbf{P}}}

\def\bbE{{\mathbb{E}}}

\def\cA{\mathcal{A}}

\def\cC{\mathcal{C}}

\def\cF{\mathcal{F}}

\def\cN{\mathcal{N}}

\def\cW{\mathcal{W}}

\def\sf0{{\mathsf{0}}}

\newcommand{\sref}[1]{{Section}~\ref{#1}}
\usepackage{epstopdf}
\usepackage{enumerate}
\usepackage{algorithmicx}
\usepackage{algorithm}
\usepackage{amsmath}
\usepackage[noend]{algpseudocode}
\usepackage{float}
\usepackage{color}
\usepackage{makeidx}
\usepackage{bbm}

\allowdisplaybreaks
\IEEEoverridecommandlockouts
\def\j{\mathrm{j}}
\begin{document}
\title{ Achievable Rates of Multi-User Millimeter Wave Systems with Hybrid Precoding}
\author{\IEEEauthorblockN{{Ahmed Alkhateeb, Robert W. Heath Jr.}
\IEEEauthorblockA{Wireless Networking and Communications Group\\
The University of Texas at Austin\\
Email: $\{$aalkhateeb, rheath$\}$@utexas.edu}
\and
\IEEEauthorblockN{Geert Leus}
\IEEEauthorblockA{Faculty of EE, Mathematics and Computer Science\\
Delft University of Technology\\
g.j.t.leus@tudelft.nl}}
}

\maketitle
\thispagestyle{empty}
\pagestyle{empty}

\begin{abstract}
Millimeter wave (mmWave) systems will likely employ large antenna arrays at both the transmitters and receivers. A natural application of antenna arrays is simultaneous transmission to multiple users, which requires multi-user precoding at the transmitter. Hardware constraints, however, make it difficult to apply conventional lower frequency MIMO precoding techniques at mmWave. This paper proposes and analyzes a low complexity hybrid analog/digital precoding algorithm for downlink multi-user mmWave systems. Hybrid precoding involves a combination of analog and digital processing that is motivated by the requirement to reduce the power consumption of the complete radio frequency and mixed signal hardware. The proposed algorithm configures hybrid precoders at the transmitter and analog combiners at multiple receivers with a small training and feedback overhead. For this algorithm, we derive a lower bound on the achievable rate for the case of single-path channels, show its asymptotic optimality at large numbers of antennas, and make useful insights for more general cases. Simulation results show that the proposed algorithm offers higher sum rates compared with analog-only beamforming, and approaches the performance of the unconstrained digital precoding solutions.
\end{abstract}

\section{Introduction} \label{sec:intro}

The large bandwidths in the mmWave spectrum make mmWave communication desirable for wireless local area networking and a candidate for future cellular systems \cite{Pi2011,Cov_Magazine,Rapp5G,ayach2013spatially,mmWave_Estimation_2013}. Achieving high quality communication links in mmWave systems requires employing large antenna arrays at both the access point or base station (BS) and the mobile stations (MS's) \cite{ayach2013spatially,Pi2011}. The transmit antenna array can be used to send data to multiple users using multi-user MIMO precoding principles. In conventional lower frequency systems, this precoding is usually performed in the digital baseband to have a better control over the entries of the precoding matrix. Unfortunately, the high cost and power consumption of mixed signal components make fully digital baseband precoding difficult at mmWave systems \cite{ayach2013spatially}. Further, constructing precoding matrices is usually based on complete channel state information, which is hard to acquire in mmWave systems \cite{mmWave_Estimation_2013}. Therefore, mmWave-specific multi-user MIMO precoding algorithms are required.

In single-user mmWave systems, analog beamforming, which controls the phase of the signal transmitted at each antenna via a network of analog phase shifters and is implemented in the radio frequency (RF) domain, is the defacto approach for implementing beamforming \cite{Wang1,brady2013beamspace}. The phase shifters might be digitally controlled with only quantized phase values. In \cite{Wang1}, adaptive beamforming algorithms and multi-resolution codebooks were developed by which the transmitter and receiver jointly  design their analog beamforming vectors. In \cite{brady2013beamspace}, beamspace MIMO was introduced for large antenna systems in which discrete Fourier transform beamforming vectors are used to asymptotically maximize the received signal power. The RF hardware constraints such as the availability of only quantized and constant modulus phase shifters, however,  limit the ability to make sophisticated processing using analog-only beamforming, for example to manage interference between users. To multiplex several data streams and perform more accurate beamforming, hybrid precoding was proposed \cite{ayach2013spatially,mmWave_Estimation_2013}, where the processing is divided between the analog and digital domains. In \cite{ayach2013spatially},\cite{mmWave_Estimation_2013}, the sparse nature of the mmWave channels was exploited to develop low-complexity hybrid precoding algorithms for single-user channels that support a limited number of streams \cite{Rapp5G}. The digital precoding layer of hybrid precoding gives more freedom in designing the precoders, compared to the analog-only solution, and seems like a promising framework on which to build multi-user MIMO mmWave precoders.

In this paper, we propose a low-complexity yet efficient hybrid analog/digital precoding algorithm for downlink multi-user mmWave systems. The proposed algorithm depends on the known (but arbitrary) array geometry and incurs a low training and feedback overhead. Our model assumes that the MS's employ analog-only combining while the BS performs hybrid analog/digital precoding. We provide a lower bound on the achievable rate using the proposed algorithm for the case of single-path channels, which is relevant for mmWave systems, and provides useful insights for more general settings. The proposed algorithm and performance bounds are also evaluated by simulations and compared with analog-only beamforming solutions. The results indicate that the proposed algorithm performs very well in mmWave systems thanks to the sparse nature of the channel and the large number of antennas used by the BS and MS's.
An extended journal version of this work was submitted recently\cite{alkhateeb2014limited}, which provides an analysis of the proposed algorithm in more general channel settings, and characterizes the rate-loss when a limited feedback exists between the BS and MS's.

We use the following notation: $\bA$ is a matrix, $\ba$ is a vector, $a$ is a scalar, and $\cA$ is a set. $\|\bA \|_F$ is the Frobenius norm of $\bA$, whereas $\bA^\mathrm{T}$, $\bA^*$, $\bA^{-1}$, are its transpose, Hermitian, and inverse, respectively. $\bbE\left[\cdot\right]$ denotes expectation.
\section{System Model} \label{sec:Model}
\begin{figure} [t]
\centerline{
\includegraphics[width=1\columnwidth]{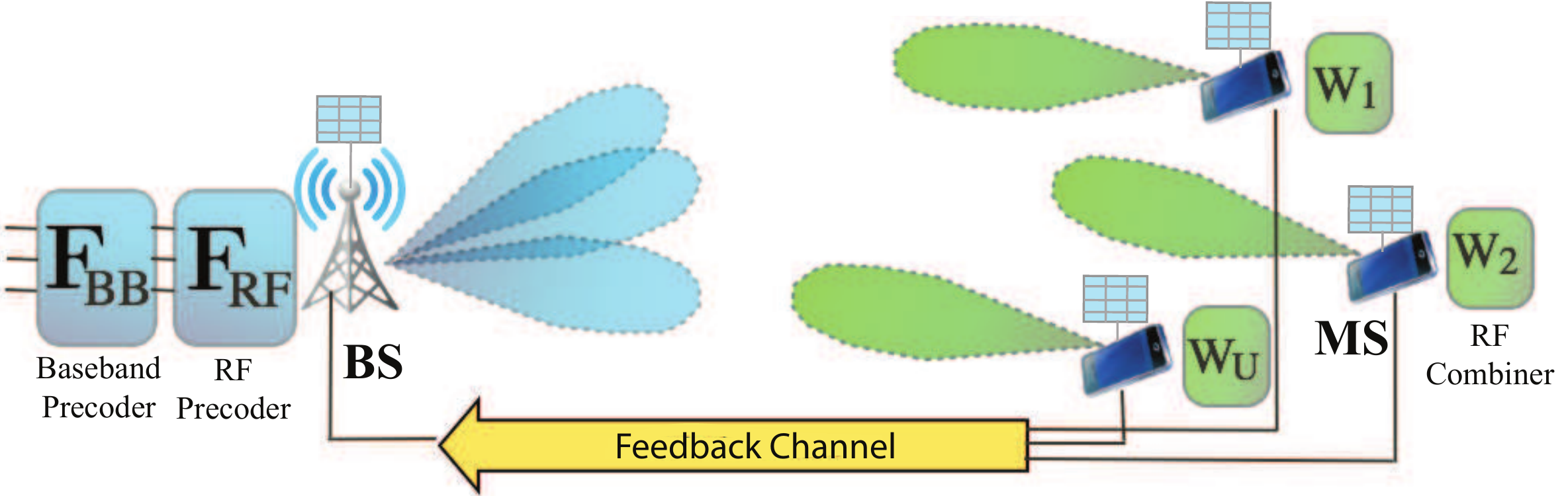}
}
\caption{A multi-user mmWave downlink system model, in which a BS uses hybrid analog/digital precoding and a large antenna array to serve $U$ MSs. Each MS employs analog-only combining to receive the signal. }
\label{fig:LimitedFB}
\end{figure}

Consider the  mmWave system shown in \figref{fig:LimitedFB}. A BS  with $N_\mathrm{BS}$ antennas and $N_\mathrm{RF}$ RF chains is assumed to communicate with $U$ MS's. Each MS is equipped with $N_\mathrm{MS}$ antennas as depicted in \figref{fig:Hybrid}. We focus on the multi-user beamforming case in which the BS communicates with every MS via \textit{only one stream}. Therefore, the total number of streams is $N_\mathrm{S}=U$. Further, we assume that the maximum number of users that can be simultaneously served by the BS equals the number of BS RF chains, i.e., $U \leq N_\mathrm{RF}$. This is motivated by the spatial multiplexing gain of the described multi-user hybrid precoding system, which is limited by $\min\left(N_\mathrm{RF},U\right)$ for $N_\mathrm{BS} > N_\mathrm{RF}$. For simplicity, we will also assume that the BS will use $U$ out of the $N_\mathrm{RF}$ available RF chains to serve the $U$ users.

On the downlink, the BS applies a $U \times U$ baseband precoder $\bF_\mathrm{BB}=\left[\bff_1^\mathrm{BB}, \bff_2^\mathrm{BB}, ..., \bff_U^\mathrm{BB}\right]$ followed by an $N_\mathrm{BS} \times U$ RF precoder, $\bF_\mathrm{RF}=\left[\bff_1^\mathrm{RF}, \bff_2^\mathrm{RF}, ..., \bff_U^\mathrm{RF}\right]$. The sampled transmitted signal is therefore
\begin{equation}
\bx=\bF_\mathrm{RF} \bF_\mathrm{BB} \bs,
\end{equation}
where $\bs=[s_1, s_2, ..., s_U]^\mathrm{T}$ is the $U \times 1$ vector of transmitted symbols, such that $\bbE\left[\bs\bs^*\right] = \frac{P}{U}\bI_U$, and $P$ is the average total transmitted power. Since $\bF_\mathrm{RF}$ is implemented using quantized analog phase shifters, $\left[\bF_\mathrm{RF}\right]_{m,n}= \frac{1}{\sqrt{N_\mathrm{BS}}} e^{\j \phi_{m,n}}$, where $\phi_{m.n}$ is a quantized angle, and the factor of  $\frac{1}{\sqrt{N_\mathrm{BS}}}$ is for power normalization.

For simplicity, we adopt a narrowband block-fading channel model as in \cite{ayach2013spatially,brady2013beamspace,mmWave_Estimation_2013} in which the $u$th MS observes the received signal as
\begin{equation}
\br_{u}=\bH_{u}\sum_{n=1}^{U}{\bF_\mathrm{RF} \bff^\mathrm{BB}_{n} s_{n}} + \bn_{u},
\label{eq:received_signal}
\end{equation}
where $\bH_{u}$ is the $N_\mathrm{MS} \times N_\mathrm{BS}$ matrix that represents the mmWave channel between the BS and the $u$th MS, and $\bn_{u} \sim \cN (\boldsymbol{0}, \sigma^2 \bI )$ is a Gaussian noise vector.

At the $u$th MS, the RF combiner $\bw_u$ is used to process the received signal $\br_u$ to produce the scalar
\begin{equation}
y_u= \bw_u^*  \bH_{u}  \sum_{n=1}^{U}{\bF_\mathrm{RF} \bff^\mathrm{BB}_{n} s_{n}} + \bw_u^* \bn_u,
\label{eq:combined_signal}
\end{equation}
where $\bw_u$ has similar constraints as the RF precoders, i.e., the constant modulus and quantized angles constraints. We assume that only analog (RF) beamforming is used at the MS's as they will likely need cheaper hardware with lower power consumption.

\begin{figure} [t]
\centerline{
\includegraphics[width=1\columnwidth, height=130pt]{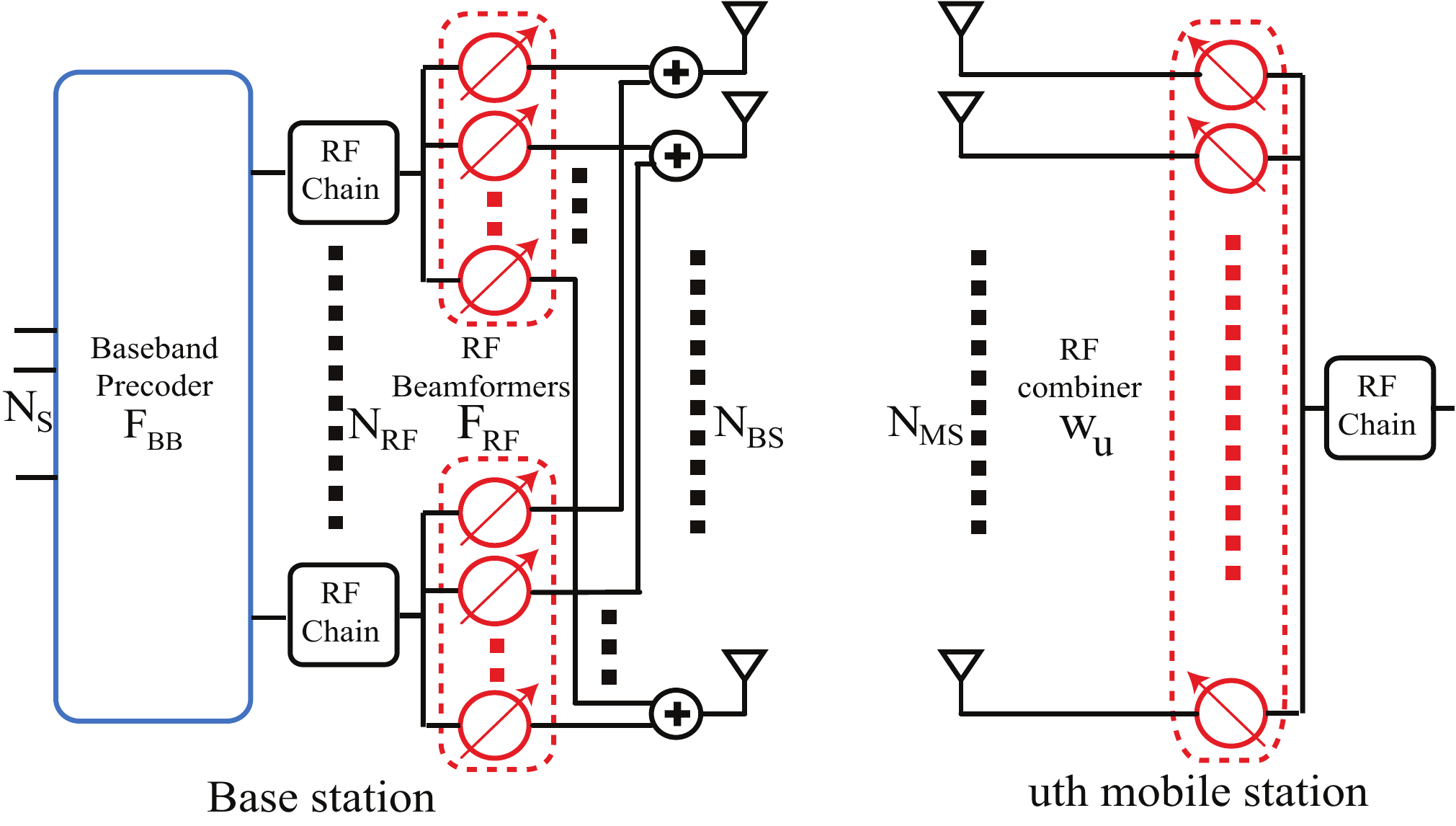}
}
\caption{A BS with hybrid analog/digital architecture communicating with the $u$th MS that employs analog-only combining.}
\label{fig:Hybrid}
\end{figure}

MmWave channels are expected to have limited scattering~\cite{Rapp5G}. To incorporate this effect, we assume a geometric channel model with $L_{u}$ scatterers for the channel of user $u$. Each scatterer is assumed to contribute a single propagation path between the BS and MS~\cite{ayach2013spatially}. The channel model can be transformed into the virtual channel model, which simplifies the generalization for larger angle spreads by incorporating spatial spreading functions \cite{brady2013beamspace}. Under this model, the channel $\bH_{u}$ can be expressed as
\begin{align}
\bH_u =\sqrt{\frac{N_\mathrm{BS} N_\mathrm{MS}} { L_{u}}} \sum_{\ell=1}^{L_{u}}\alpha_{u,\ell} \ba_\mathrm{MS}\left(\theta_{u,\ell}\right) \ba^*_\mathrm{BS} \left(\phi_{u,\ell} \right),
\label{eq:channel_model}
\end{align}
where $\alpha_{u,\ell}$ is the complex gain of the $\ell^\mathrm{th}$ path, including the path-loss, with $\mathbb{E}\left[|\alpha_{u,\ell}|^2\right]=\bar{\alpha}$. The variables $\theta_{u,\ell}$, and $\phi_{u,\ell} \in [0, 2\pi]$ are the $\ell^\mathrm{th}$ path's angles of arrival and departures (AoA/AoD) respectively. Finally, $\ba_\mathrm{BS}\left(\phi_{u,\ell}\right)$ and $\ba_\mathrm{MS}\left(\theta_{u,\ell}\right)$ are the antenna array response vectors of the BS and $u$th MS respectively. The BS and each MS are assumed to know the geometry of their antenna arrays. While the algorithms and results developed in the paper can be applied to arbitrary antenna arrays, we use uniform planar arrays (UPAs), with both azimuth and elevation components of each angle, in the simulations of \sref{sec:Results}.

\section{Problem Formulation} \label{sec:Form}
Given the received signal at the $u$th MS in \eqref{eq:received_signal}, which is then processed using the RF combiner $\bw_u$, the achievable rate of user $u$ is
\begin{equation}
R_u=\log_2\left(1+\frac{P \left|\bw_u^* \bH_u \bF_\mathrm{RF} \bff_u^\mathrm{BB}\right|^2}{P \sum_{n\neq u} { \left|\bw_u^* \bH_u \bF_\mathrm{RF} \bff_n^\mathrm{BB}\right|^2 + \sigma^2} }\right).
\end{equation}
The sum-rate of the system is then  $R_\mathrm{sum}=\sum_{u=1}^U R_u$.
Our main objective is to efficiently design the analog (RF) and digital (baseband) precoders at the BS and the analog combiners at the MS's to maximize the sum-rate of the system.

Due to the constraints on the RF hardware, such as the availability of only quantized angles for the RF phase shifters, the analog beamforming/combining vectors can take only certain values. Hence, these vectors need to be selected from finite-size codebooks. There are different models for the RF beamforming codebooks, two possible examples are
\begin{enumerate}
\item{\textbf{General quantized beamforming codebooks} Here, the codebooks are usually designed for rich channels and, therefore, attempt a uniform quantization of the space of beamforming vectors. These codebooks were commonly used in traditional MIMO systems.}
\item{\textbf{Beamsteering codebooks} The beamforming vectors, here, are spatial matched filters for the single-path channels. As a result, they have the same form of the array response vector and can be parameterized by a simple angle. Let $\cF$ represent the RF beamforming codebook with cardinality $\left|\cF\right|=N_\mathrm{Q}$. Then,  $\cF$ consists of the vectors $\ba_\mathrm{BS}\left(\frac{2 \pi k_\mathrm{Q}}{N_\mathrm{Q}}\right)$, for the variable $k_\mathrm{Q}$ taking the values $0, 1, 2$, and $N_\mathrm{Q}-1$. The RF combining vectors codebook $\cW$ can be similarly defined.}
\end{enumerate}

Motivated by the good performance of single-user hybrid precoding algorithms \cite{ayach2013spatially,mmWave_Estimation_2013} which rely on RF beamsteering vectors, and by the relatively small size of these codebooks which depend on a single-parameter quantization, we will adopt the beamsteering codebooks for the analog beamforming vectors. While the problem formulation and proposed algorithm in this paper are general for any codebook, the performance evaluation of the proposed algorithm done in Section \ref{sec:Performance} depends on the selected codebook.

If the system sum-rate is adopted as a performance metric, the precoding design problem is then to find $\bF_\mathrm{RF}^\star$, $\left\{\bff_u^{ \star \mathrm{BB}}\right\}_{u=1}^U$ and $\left\{\bw_u^\star\right\}_{u=1}^U$ that solve
\begin{align}
\begin{split}
\left\{\bF_\mathrm{RF}^\star, \left\{\bff_u^{ \star \mathrm{BB}}\right\}_{u=1}^U, \left\{\bw_u^\star\right\}_{u=1}^U\right\} & = \\
& \hspace{-120pt} \arg\max {\sum_{u=1}^{U} \log_2\left(1+\frac{P \left|\bw_u^* \bH_u \bF_\mathrm{RF} \bff_u^\mathrm{BB}\right|^2}{P \sum_{n\neq u} { \left|\bw_u^* \bH_u \bF_\mathrm{RF} \bff_n^\mathrm{BB}\right|^2 + \sigma^2} }\right)} \\
& \hspace{-117pt} \mathrm{s.t.} \hspace{23 pt}\left[\bF_\mathrm{RF}\right]_{:,u} \in \cF, u=1,2,..., U, \\
& \hspace{-79pt}  \bw_u \in \cW, u=1,2,..., U, \\
& \hspace{-79pt}  \|\bF_\mathrm{RF} \left[\bff_1^\mathrm{BB}, \bff_2^\mathrm{BB}, ..., \bff_U^\mathrm{BB} \right]\|_F^2=U.
\end{split} \label{eq:Opt}
\end{align}

The problem in \eqref{eq:Opt} is a mixed integer programming problem. Its solution requires a search over the entire $\cF^U \times \cW^U $ space of all possible $\bF_\mathrm{RF}$ and $\left\{\bw_u\right\}_{u=1}^U$ combinations. Further, the digital precoder $\bF_\mathrm{BB}$ needs to be jointly designed with the analog beamforming/combining vectors. In practice, this may require the feedback of the channel matrices $\bH_u, u=1, 2, ..., U$, or the effective channels, $\bw_u^* \bH_u \bF_\mathrm{RF}$. Therefore, the solution of \eqref{eq:Opt} requires large training and feedback overhead. Moreover, the optimal digital linear precoder is not known in general even without the RF constraints, and only iterative solutions exist \cite{Coordinated}. Hence, the direct solution of this sum-rate maximization problem is not practical.

Similar problems to \eqref{eq:Opt} have been studied before in literature, but with baseband (not hybrid) precoding and combining \cite{Coordinated, Boccardi_BD,Jindal_Comb}. The application of these algorithms in mmWave systems, however, is generally difficult due to (i) the large feedback overhead associated with the large antenna arrays at both the BS and MS's, (ii) the need for the RF beamforming/combining vectors to be taken from quantized codebooks which requires a search over both the BS and MS's beamforming vectors \cite{mmWave_Estimation_2013}, and limits the freedom in designing the combining vectors, and (iii) the convergence of iterative precoding/combining algorithms like \cite{Coordinated} has not yet been studied for hybrid precoders.

Given the practical difficulties associated with applying the   precoding/combining algorithms in mmWave systems, we propose a new mmWave-suitable multi-user MIMO beamforming algorithm in \sref{sec:MU_Alg}. Our proposed algorithm is developed to achieve a good performance compared with the solution of \eqref{eq:Opt}, while requiring low training and feedback overhead.

\section{Two-stage Multi-user Hybrid Precoding} \label{sec:MU_Alg}

\begin{algorithm} [!t]                     
\caption{Multi-user Hybrid Precoders Algorithm}          
\label{alg:MU_Precoding}                           
\begin{algorithmic} 
    \State \textbf{Input:} $\cF$ and $\cW$, BS and MS RF beamforming codebooks
    \State \textbf{First stage:} Single-user RF beamforming/combining design
    \State  For each MS $u, u=1, 2, ..., U$
    \State \hspace{20pt} The BS and MS $u$ select  $\bv_u^\star$ and $\bg_u^\star$ that solve
    \State \hspace{80pt} $\left\{\bg_u^\star, \bv_u^\star\right\}=\displaystyle{\operatorname*{\arg max}_{ \substack{\forall \bg_u \in \cW \\ \forall \bv_u \in \cF }}}{\|\bg_u^* \bH_u \bv_u\|} $
    \State \hspace{20pt} MS $u$ sets $\bw_u=\bg_u^\star$
    \State BS sets $\bF_\mathrm{RF}=[\bv_1^\star, \bv_2^\star, ..., \bv_U^\star]$
    \State \textbf{Second stage:} Multi-user digital precoding design
    \State For each MS $u, u=1, 2, ..., U$
    \State \hspace{20pt} MS $u$ estimates its effective channel $\overline{\bh}_u=\bw_u^* \bH_u \bF_\mathrm{RF}$
    \State \hspace{20pt} MS $u$ feeds back its effective channel $\overline{\bh}_u$ to the BS
    \State BS designs $\bF_\mathrm{BB}= {\overline{\bH}}^*\left(\overline{\bH} {\overline{\bH}}^*\right)^{-1}$, $\overline{\bH}=\left[\overline{\bh}_1^\mathrm{T},..., \overline{\bh}_U^\mathrm{T}\right]^\mathrm{T}$
    \State $\bff_u^\mathrm{BB}= \frac{\bff_u^\mathrm{BB}}{\left\|\bF_\mathrm{RF} \bff_u^\mathrm{BB}\right\|_F}, u=1, 2, ..., U$
    \end{algorithmic}
\end{algorithm}

The additional challenge in solving (7), beyond the usual coupling between precoders and combiners \cite{Coordinated, Boccardi_BD,Jindal_Comb}, is the splitting of the precoding operation into two different domains, each with different constraints. To overcome that, we propose Algorithm \ref{alg:MU_Precoding} that divides the calculation of the precoders into two stages. In the first stage, the BS RF precoder and the MS RF combiners are jointly designed to maximize the desired signal power of each user, neglecting the resulting interference among users. In the second stage, the BS digital precoder is designed to manage the multi-user interference. The operation of Algorithm \ref{alg:MU_Precoding} can be summarized as follows.

\textbf{In the first stage:} The BS and each MS $u$ design the RF beamforming and combining vectors, $\bff^\mathrm{RF}_u$ and $\bw_u$, to maximize the desired signal power for user $u$, and neglecting the other users' interference. As this is the typical single-user RF beamforming design problem, efficient beam training algorithms developed for single-user systems that do not require explicit channel estimation and have a \textit{low training overhead}, can be used to design the RF beamforming/combining vectors. For example, \cite{Wang1,mmWave_Estimation_2013} developed multi-resolution codebooks to adaptively refine the beamforming/combining vectors, and hence avoid the exhaustive search complexity.

\textbf{In the second stage:} The BS trains the effective channels, $\overline{\bh}_u=\bw_u^* \bH_u \bF_\mathrm{RF}, u=1, 2, ..., U$, with the MS's. Note that the dimension of each effective channel vector is $U \times 1$ which is \textit{much less} than the original channel matrix.  Then, each MS $u$ feeds its effective channel back to the BS, which designs its zero-forcing digital precoder based on these effective channels. Thanks to the narrow beamforming and the sparse mmWave channels, the effective channels are expected to be well-conditioned which makes adopting a multi-user digital beamforming strategy like zero-forcing capable of achieving near-optimal performance as will be shown in \sref{sec:Performance}.

\section{Achievable Rate with Single-Path Channels} \label{sec:Performance}
The analysis of hybrid precoding is non-trivial due to the coupling between analog and digital precoders. Therefore, we will study the performance of the proposed algorithm in the case of single-path channels. This  case is of special interest as mmWave channels are likely to be sparse, i.e., only a few paths exist \cite{Rapp5G}. Further, the analysis of this special case will give useful insights into the performance of the proposed algorithms in more general settings.

Next, we will  characterize a lower bound on the achievable rate by each MS when Algorithm \ref{alg:MU_Precoding} is used to design the hybrid precoders at the BS and RF combiners at the MS's. Consider the BS and MS's with the system and channels described in \sref{sec:Model} with the following assumptions:
\begin{assumption}
All channels are single-path, i.e., $L_u=1$, $u=1, 2, ..., U$. For ease of exposition, we will omit the subscript $\ell$ in the definition of the channel parameters in \eqref{eq:channel_model}.
\end{assumption}
\begin{assumption}
The RF beamforming and combining vectors $\bff_u^\mathrm{RF}$ and $\bw_u, u=1, 2, ..., U$, $u=1,2, ..., U$  are beamsteering vectors with continuous angles.
\end{assumption}
\begin{assumption}
The BS perfectly knows the effective channels $\overline{\bh}_u, u=1, 2, ..., U$.
\end{assumption}

In the first stage of Algorithm \ref{alg:MU_Precoding}, the BS and each MS $u$ find $\bv_u^\star$ and $\bg_u^\star$ that solve
\begin{equation}
\left\{\bg_u^\star, \bv_u^\star\right\}=\displaystyle{\operatorname*{\arg max}_{ \substack{\forall \bg_u \in \cW \\ \forall \bv_u \in \cF }}}{\|\bg_u^* \bH_u \bv_u\|}.
\label{eq:Stage1}
\end{equation}

As the channel $\bH_u$ has only one path, and given the continuous beamsteering capability assumption, the optimal RF precoding and combining vectors will be $\bg_u^\star=\ba_\mathrm{MS}(\theta_u)$, and $\bv_u^\star=\ba_\mathrm{BS}(\phi_u)$. Consequently, the MS sets $\bw_u=\ba_\mathrm{MS}(\theta_u)$ and the BS takes $\bff_u^\mathrm{RF}=\ba_\mathrm{BS}(\phi_u)$. If we let the $N_\mathrm{BS} \times U$  matrix $\bA_\mathrm{BS}$  gather the BS array response vectors associated with the $U$ AoDs, i.e.,  $\bA_\mathrm{BS}=\left[\ba_\mathrm{BS}\left(\phi_1\right), \ba_\mathrm{BS}\left(\phi_2\right), ..., \ba_\mathrm{BS}\left(\phi_U\right) \right]$, we can then write the BS RF beamforming matrix, including the beamforming vectors of the $U$ users as $\bF_\mathrm{RF}=\bA_\mathrm{BS}$.

The effective channel for user $u$ after designing the RF precoders and combiners is
\begin{equation}
\begin{split}
\overline{\bh}_u&=\bw_u \bH_u \bF_\mathrm{RF}\\
&=\sqrt{N_\mathrm{BS} N_\mathrm{MS}} \alpha_u \ba^*_\mathrm{BS}\left(\phi_u\right) \bF_\mathrm{RF}.
\end{split}
\end{equation}

Now, defining $\overline{\bH}=[\overline{\bh}_1^\mathrm{T}, \overline{\bh}_2^\mathrm{T}, ..., \overline{\bh}_U^\mathrm{T}]^\mathrm{T}$, and given the design of $\bF_\mathrm{RF}$, we can write the effective channel matrix $\overline{\bH}$ as
\begin{equation}
\overline{\bH}=\bD \bA_\mathrm{BS}^* \bA_\mathrm{BS}, \label{eq:eff_channel}
\end{equation}
where $\bD$ is a $U \times U$ diagonal matrix, $\left[\bD\right]_{u,u}=\sqrt{N_\mathrm{BS} N_\mathrm{MS}} \alpha_u$.

Based on this effective channel, the BS zero-forcing digital precoder is defined as
\begin{equation}
\bF_\mathrm{BB}=\overline{\bH}^* \left(\overline{\bH} \overline{\bH}^*\right)^{-1} \boldsymbol\Lambda,
\end{equation}
where $\boldsymbol\Lambda$ is a diagonal matrix with the diagonal elements adjusted to satisfy the precoding power constraints $\left\|\bF_\mathrm{RF} \bff_u^\mathrm{BB}\right\|^2=1, u=1, 2, ..., U$. The diagonal elements of $\boldsymbol\Lambda$ are then equal to [See Appendix A in \cite{alkhateeb2014limited} for a derivation]
\begin{equation}
{\boldsymbol\Lambda}_{u,u}=\sqrt{\frac{N_\mathrm{BS} N_\mathrm{MS} }{ \left(\bA_\mathrm{BS}^* \bA_\mathrm{BS}\right)^{-1}_{u,u}}}\left|\alpha_u\right|, u=1, 2, ..., U.
\end{equation}

Note that this $\boldsymbol\Lambda$ is different than the traditional digital zero-forcing precoder due to the different power constraints in the hybrid analog/digital architecture.

The achievable rate for user $u$ is then
\begin{align}
\begin{split}
R_u&=\log_2 \left(1+\frac{\mathsf{SNR}}{U} \left|\overline{\bh}_u^* \bff_u^\mathrm{BB}\right|^2\right),\\
&=\log_2\left(1+\frac{\mathsf{SNR}}{ U} \frac{N_\mathrm{BS} N_\mathrm{MS} \left|\alpha_u\right|^2 }{\left(\bA_\mathrm{BS}^* \bA_\mathrm{BS}\right)^{-1}_{u,u}}\right). \label{eq:Rate_u}
\end{split}
\end{align}

To bound this rate, the following lemma, which characterizes a useful property of the matrix $\bA_\mathrm{BS}^* \bA_\mathrm{BS}$, can be used.
\begin{lemma}
Let $\bA_\mathrm{BS}=\left[\ba_\mathrm{BS}\left(\phi_1\right), \ba_\mathrm{BS}\left(\phi_2\right), ..., \ba_\mathrm{BS}\left(\phi_U\right) \right]$, with the angles $\phi_u, u=1, 2, ..., U$ taking continuous values in $[0, 2 \pi]$, then the matrix $\bP=\bA_\mathrm{BS}^* \bA_\mathrm{BS}$ is positive definite almost surely.
\end{lemma}
\begin{proof}
Let the matrix $\bP = \bA_\mathrm{BS}^* \bA_\mathrm{BS}$, then for any non-zero complex vector $\bz \in \cC^{N_\mathrm{BS}}$, it follows that $\bz^* \bP \bz = \|\bA_\mathrm{BS} \bz\|_2^2 \geq 0$. Hence, the matrix $\bP$ is positive semi-definite. Further, if the vectors $\ba_\mathrm{BS}\left(\phi_1\right), \ba_\mathrm{BS}\left(\phi_2\right)$, $ ..., \ba_\mathrm{BS}\left(\phi_U\right)$ are linearly independent, then for any non-zero complex vector $\bz$, $\bA_\mathrm{BS} \bz \neq 0$, and the matrix $\bP$ is positive definite. To show that, consider any two vectors $\ba_\mathrm{BS}\left(\phi_u\right), \ba_\mathrm{BS}\left(\phi_n\right)$. These vectors are linearly dependent if and only if $\phi_u=\phi_n$. As the probability of this event equals zero when the AoDs $\phi_u$ and $\phi_n$ are selected independently from a continuous distribution, the matrix $\bP$ is positive definite with probability one.
\end{proof}

Now, using the Kantorovich inequality \cite{kantorovich1948functional}, we can bound the diagonal entries of the matrix $\left(\bA_\mathrm{BS}^* \bA_\mathrm{BS}\right)^{-1}$ using the following lemma from \cite{bai1996bounds}.
\begin{lemma}
For any $n \times n$ Hermitian and positive definite matrix $\bP$ with the ordered eigenvalues satisfying $0<\lambda_\mathrm{min} \leq \lambda_2 \leq ... \leq \lambda_\mathrm{max}$, the element $\left(\bP \right)^{-1}_{u,u}, u=1, 2, ..., n$ satisfies
\begin{equation}
\left(\bP \right)^{-1}_{u,u} \leq \frac{1}{4 [\bP]_{u,u}}\left(\frac{\lambda_\mathrm{max}\left(\bP\right)}{\lambda_\mathrm{min}\left(\bP\right)}+\frac{\lambda_\mathrm{min}\left(\bP\right)}{\lambda_\mathrm{max}\left(\bP\right)}+2\right).
\end{equation}
\label{lemma3}
\end{lemma}

We also note that for the matrix $\bA_\mathrm{BS}^* \bA_\mathrm{BS}$, we have $\left(\bA_\mathrm{BS}^* \bA_\mathrm{BS}\right)_{u,u}=1$, $\lambda_\mathrm{min}\left(\bA_\mathrm{BS}^* \bA_\mathrm{BS}\right)=\sigma_\mathrm{min}^2\left(\bA_\mathrm{BS}\right)$, and $\lambda_\mathrm{max}\left(\bA_\mathrm{BS}^* \bA_\mathrm{BS}\right)=\sigma_\mathrm{max}^2\left(\bA_\mathrm{BS}\right)$, where $\sigma_{\mathrm{max}}(\bA_\mathrm{BS})$ and $\sigma_{\mathrm{min}}(\bA_\mathrm{BS})$ are the maximum and minimum singular values, respectively. Using this note and Lemma \ref{lemma3}, and defining
\begin{equation}
G\left(\left\{\phi_u\right\}_{u=1}^U\right)= 4 \left(\frac{\sigma_{\mathrm{max}}^2\left(\bA_\mathrm{BS}\right)}{\sigma_{\mathrm{min}}^2\left(\bA_\mathrm{BS}\right)}+\frac{\sigma_{\mathrm{min}}^2\left(\bA_\mathrm{BS}\right)}{\sigma_{\mathrm{max}}^2\left(\bA_\mathrm{BS}\right)}+2\right)^{-1},
\end{equation}
we can bound the achievable rate of user $u$ in \eqref{eq:Rate_u} as
\begin{equation}
R_u \geq \log_2\left(1+ \frac{\mathsf{SNR}}{U}{N_\mathrm{BS} N_\mathrm{MS} \left|\alpha_u\right|^2} G\left(\left\{\phi_u\right\}_{u=1}^U\right) \right), \label{eq:lower_perfect}
\end{equation}

\begin{figure} [t]
\centerline{
\includegraphics[width=1.1\columnwidth]{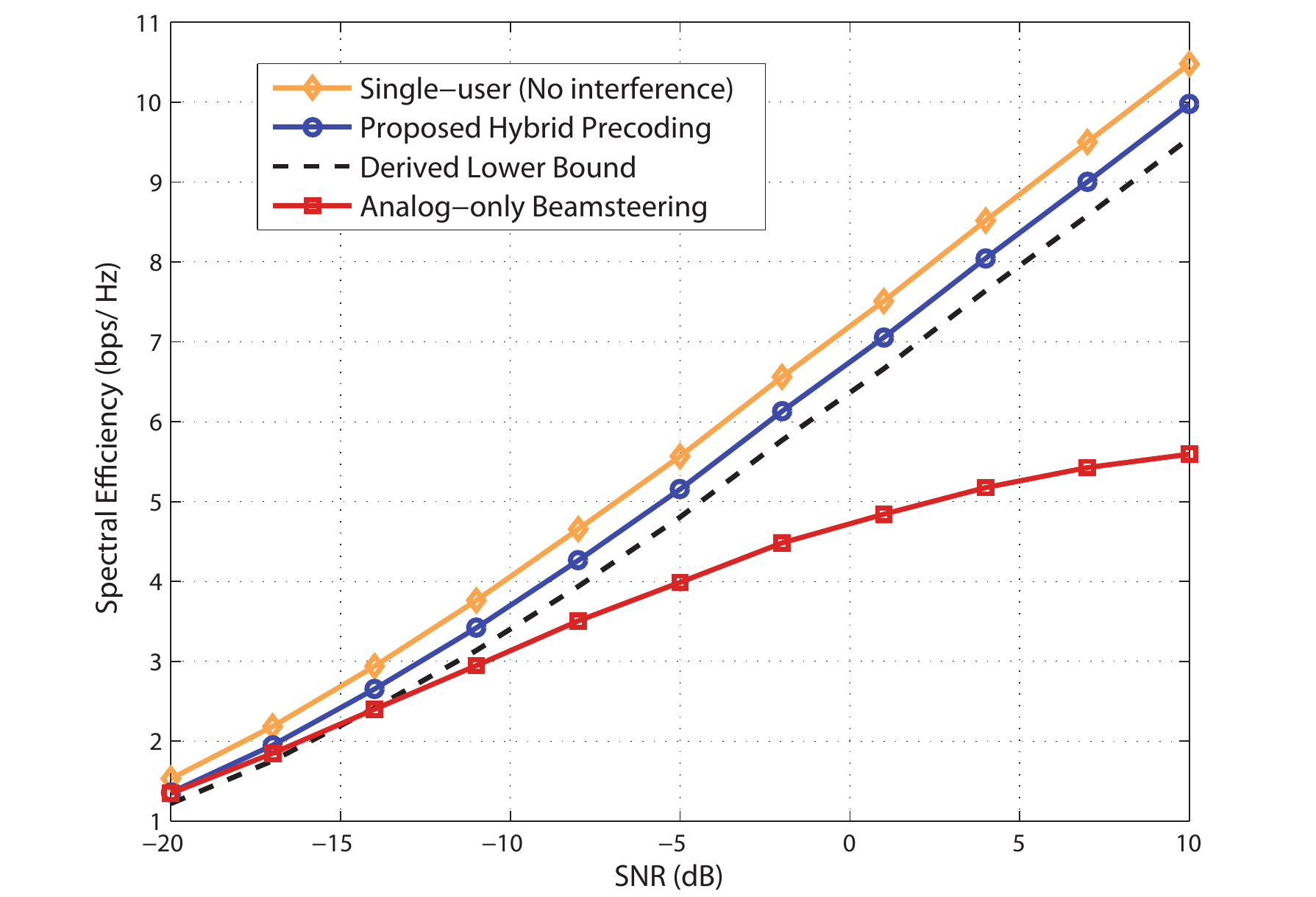}
}
\caption{Achievable rates per-user using the hybrid  precoding and beamsteering algorithms with perfect channel knowledge.}
\label{fig:Perfect_SNR}
\end{figure}

In addition to characterizing a lower bound on the rates achieved by the proposed hybrid analog/digital precoding algorithm, the bound in \eqref{eq:lower_perfect}  separates the dependence on the channel gains $\alpha_u$, and the AoDs $\phi_u, u=1, 2, ..., U$ which can be used to claim the optimality of the proposed algorithm in some cases and to give useful insights into the  gain of the proposed algorithm over analog-only beamsteering solutions. This is illustrated in the following results.

\begin{proposition}
Let $\mathring{R}_u = \log_2\left(1+\frac{\mathsf{SNR}}{ U}N_\mathrm{BS} N_\mathrm{MS} \left|\alpha_u\right|^2\right)$ denote the single-user rate, and let $R^\mathrm{BS}_u$ denote the rate achieved by user $u$ when the BS employs analog-only beamsteering designed based on the first stage of Algorithm \ref{alg:MU_Precoding}, i.e., $R^\mathrm{BS}_u = \log_2\left(1+\frac{\frac{SNR}{ U}  N_\mathrm{BS} N_\mathrm{MS} \left|\alpha_u\right|^2 }{\frac{SNR}{ U}  N_\mathrm{BS} N_\mathrm{MS} \left|\alpha_u\right|^2 \sum_{n \neq u} \left|\beta_{u,n}\right|^2+1}\right)$, with $\beta_{u,n}=\ba^*_\mathrm{BS}\left(\phi_u\right)\ba_\mathrm{BS}\left(\phi_n\right)$. When Algorithm \ref{alg:MU_Precoding} is used to design the hybrid precoders and RF combiners described in \sref{sec:Model}, and given the Assumptions 1-3, the achievable rate by any user $u$ satisfies
\begin{enumerate}
\item{$\bbE\left[ \mathring{R}_u- R_u\right]\leq K\left(N_\mathrm{BS},U\right)$},
\item{$\lim_{N_\mathrm{BS}\rightarrow\infty}{R_u}=\mathring{R}_u$ with probability one},
\item{$\lim_{N_\mathrm{MS}\rightarrow \infty}\bbE\left[R_u-R^\mathrm{BS}_u\right]= \infty$},
\end{enumerate}
where $K\left(N_\mathrm{BS},U\right)$ is a constant whose value depends only on $N_\mathrm{BS}$ and $U$.
\label{Prop1}
\end{proposition}

\begin{proof}
Please refer to Appendices  B and C in \cite{alkhateeb2014limited}.
\end{proof}

Proposition \ref{Prop1} indicates that the average achievable rate of any user $u$ using the proposed low-complexity precoding/combining algorithm grows with the same slope of that of the single-user rate at high SNR, and stays within a constant gap from it. This gap, $K\left(N_\mathrm{BS}, U\right)$, depends only on the number of users and the number of BS antennas. As the number of BS antennas increases, the gap between the achievable rate using Algorithm \ref{alg:MU_Precoding} and the single-user rate decreases, and approaches zero at infinite antenna numbers. One important note is that this gap does not depend on the number of MS antennas, which is contrary to the analog-only beamsteering, given by the first stage only of Algorithm \ref{alg:MU_Precoding}. This leads to the third part of the proposition. This implies that multi-user interference management is still important at mmWave systems even with large numbers of antennas at the BS and MS's, and perfect alignment. Note also that this is not the case when the number of BS antennas goes to infinity as it can be easily shown that the performance of RF beamsteering alone becomes optimal in this case.

\section{Simulation Results}\label{sec:Results}

\begin{figure} [t]
\centerline{
\includegraphics[width=1.1\columnwidth]{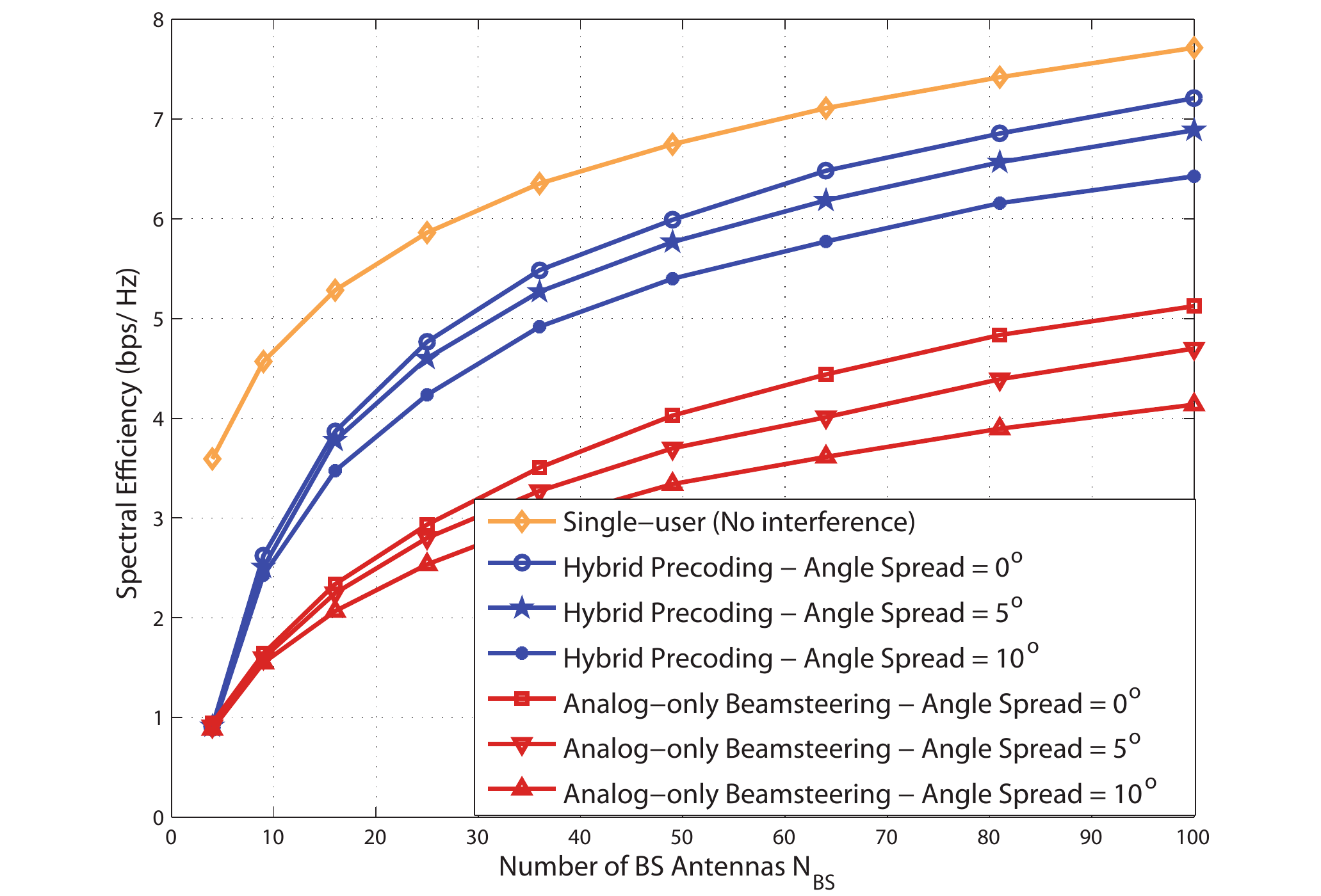}
}
\caption{Achievable rates per-user using the hybrid precoding and beamsteering algorithms for different values of angle spread.}
\label{fig:AS}
\end{figure}

In this section, we evaluate the performance of the proposed hybrid analog/digital precoding algorithm and the presented bounds using numerical simulations.

First, we compare the achievable rates without quantization loss in \figref{fig:Perfect_SNR}, where we consider the system model in \sref{sec:Model} with a BS employing an $8 \times 8$ UPA with $4$ MS's, each having a $4 \times 4$ UPA. The channels are single-path and Rayleigh distributed, the azimuth and elevation AoAs/AoDs are assumed to be uniformly distributed in $[0, 2 \pi ]$ and $[- \frac{\pi}{2}, \frac{\pi}{2}]$, respectively. The rate achieved by the proposed algorithm is compared with single-user and beamsteering rates. The figure shows that the performance of hybrid precoding is very close to the single-user rate thanks to canceling the residual multi-user interference. The gain over beamsteering increases with SNR as the beamsteering rate starts to be interference limited. The tightness of the derived lower bound is also shown.

To evaluate the performance of the proposed hybrid precoding algorithm in more general channel settings and system imperfections, we consider a 3-cluster channel model in \figref{fig:AS}. Each cluster has 6 rays with Laplacian distributed AoAs/AoDs \cite{ayach2013spatially}. The phase shifters at the BS and MS's are assumed to be quantized with 6 bits and 4 bits, respectively. The figure shows that when angle-spread and system imperfections exist, the performance of hybrid precoding is still within a small gap from the single-user rate. Further, it provides a good gain over analog-only solutions.

Finally, \figref{fig:Coverage} evaluates the performance of the proposed algorithm in a mmWave cellular setup including inter-cell interference, which is not explicitly incorporated into our designs. In this setup, BS's and MS's are assumed to be spatially distributed according to a Poisson point process with MS's densities 30 times the BS densities. The channels between the BS's and MS's are single-path and each link is determined to be line-of-sight or non-line-of-sight based on the blockage model in \cite{Cov_Magazine}. Each MS is associated to the BS with less path-loss and the BS randomly selects $ n=2,..,5$ users of those associated to it to be simultaneously served. BS's are assumed to have $8 \times 8$ UPAs and MS's are equipped with $4 \times 4$ UPAs. All UPA's are vertical, elevation angles are assumed to be fixed at $\pi/2$, and azimuth angles are uniformly distributed in $[0, 2 \pi]$. \figref{fig:Coverage} shows the per-user coverage probability defined as $\mathcal{P}\left(\mathrm{R_u \geq \eta}\right)$, where $\eta$ is an arbitrary threshold. This figure illustrates that hybrid precoding has a reasonable coverage gain over analog-only beamsteering, especially when large numbers of users are simultaneously served, thanks to the interference management capability of hybrid precoding.
\section{Acknowledgement}\label{sec:Ack}
This work is supported in part by the National Science Foundation under Grant No. 1218338 and 1319556, and by a gift from Huawei Technologies, Inc.

\section{Conclusions}\label{sec:conclusion}
In this paper, we proposed a low-complexity hybrid analog/digital precoding algorithm for downlink multi-user mmWave systems inspired by the sparse nature of the channel and the large number of deployed antennas. For the single-path channels, we characterized the achievable rates of the proposed algorithm and illustrated its asymptotic optimality. Compared with analog-only beamforming solutions, we showed that higher sum-rates can be achieved using hybrid precoding. The results indicate that interference management in multi-user mmWave systems is still important even when the number of antennas is very large. Simulation results verify these conclusions and show that they can be extended to more general multi-path mmWave channel settings.
\begin{figure}[t]
\centerline{
\includegraphics[width=1.1\columnwidth]{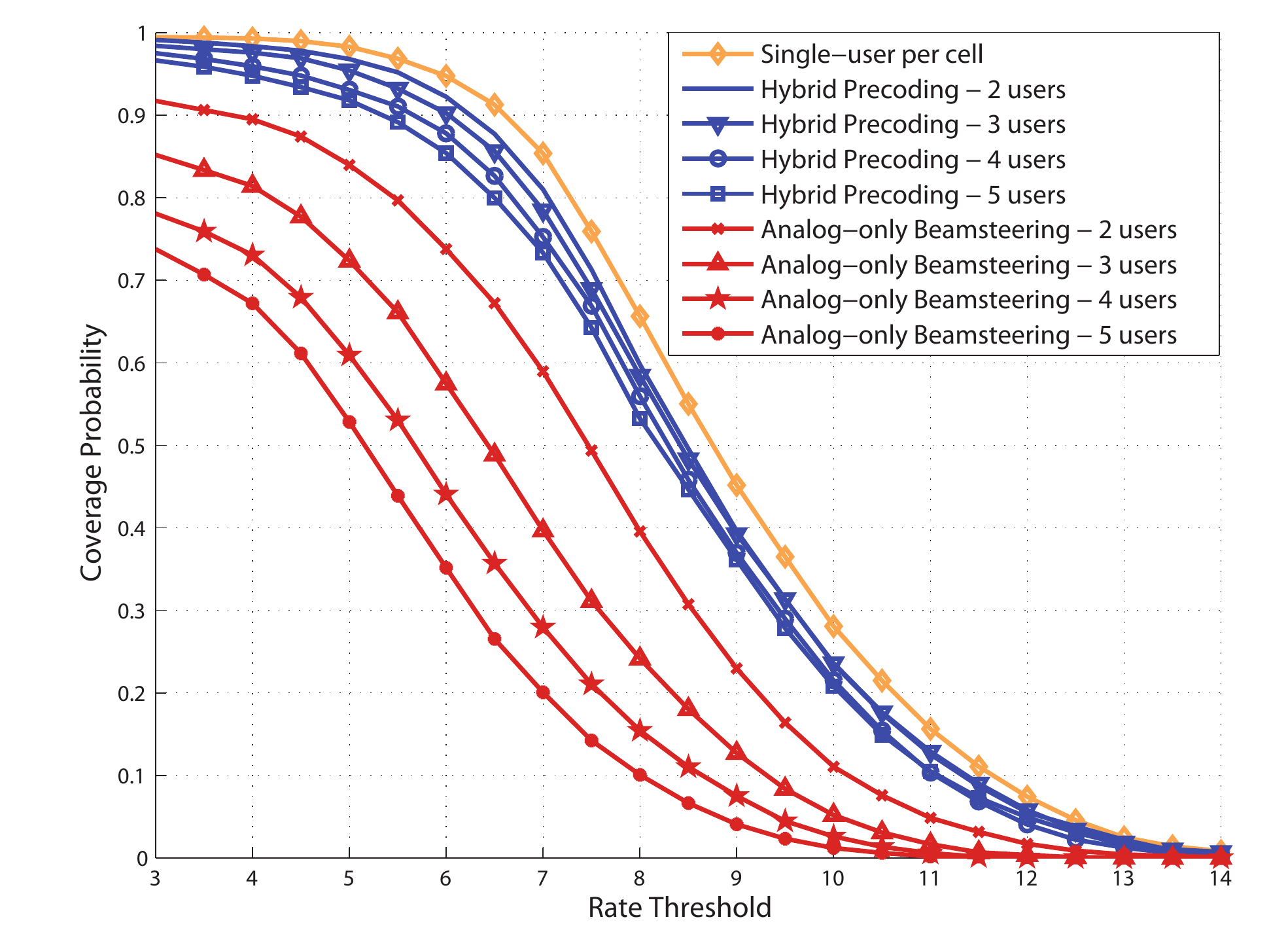}
}
\caption{Coverage probability of  hybrid precoding compared with single-user per cell and analog-only beamsteering solutions. The figure shows the per-user performance with different numbers of users per cell.}
\label{fig:Coverage}
\end{figure}


\begin{small}

\end{small}

\end{document}